\documentclass[11pt]{article}
\usepackage{fullpage}
\usepackage[numbers]{natbib}

\usepackage{amsfonts}
\usepackage{algorithmic}
\usepackage[boxed]{algorithm}
\usepackage{amsmath,amssymb}
\usepackage{amsmath,amsthm,amssymb}
\usepackage{latexsym}
\usepackage{epic}
\usepackage{epsfig}
\usepackage{amscd}
\usepackage{url}
\usepackage{verbatim}
\usepackage{setspace}

\newtheorem{theorem}{Theorem}[section]
 \newtheorem{corollary}[theorem]{Corollary}
 \newtheorem{lemma}[theorem]{Lemma}

 \newtheorem{definition}[theorem]{Definition}
 \newtheorem{remark}[theorem]{Remark}
 
 \newtheorem{observation}[theorem]{Observation}
 \newtheorem{fact}[theorem]{Fact}

 \newtheorem{note}[theorem]{Note}

\newcommand{\cross}{\times}

\newcommand{\set}[1]{\left\{ #1 \right\}}
\newcommand{\union}{\cup}
\newcommand{\intersect}{\cap}

\newcommand{\sm}{\setminus}

\renewcommand{\hat}{\widehat}
\renewcommand{\tilde}{\widetilde}

\def\ex{\qopname\relax n{E}}
\def\min{\qopname\relax n{min}}
\def\max{\qopname\relax n{max}}
\def\maxtwo{\qopname\relax n{max2}}

\def\argmax{\qopname\relax n{argmax}}

\def\Pr{\qopname\relax n{\mathbf{Pr}}}
\def\Ex{\qopname\relax n{\mathbf{E}}}

\newcommand{\RR}{\mathbb{R}}
\newcommand{\RRp}{\RR_+}

\def\D{\mathcal{D}}

\def\F{\mathcal{F}}

\def\I{\mathcal{I}}

\def\P{\mathcal{P}}

\def\S{\mathcal{S}}

\def\sse{\subseteq}

\newcommand{\eat}[1]{}

\newcommand{\INPUT}{\item[\textbf{Input:}]}
\newcommand{\OUTPUT}{\item[\textbf{Output:}]}

\newenvironment{lp*}{\begin{equation*}  \begin{array}{lll}}{\end{array}\end{equation*}}

\begin{document}

\title{Constrained Signaling in Auction Design}

\author{Shaddin Dughmi \thanks{University of Southern California} \and
 Nicole Immorlica \thanks{Microsoft Research.  This work was partially supported by NSF CAREER Grant CCF-1055020, the Alfred P.\ Sloan Research Fellowship, and the Microsoft New Faculty Fellowship.} \and
 Aaron Roth \thanks{University of Pennsylvania. This work was partially supported by an NSF CAREER Grant and NSF Grant CCF-1101389}
}

\maketitle

\begin{abstract}\small\baselineskip=9pt
We consider the problem of an auctioneer who faces the task of selling a good (drawn from a known distribution) to a set of buyers, when the auctioneer does not have the capacity to describe to the buyers the exact identity of the good that he is selling. Instead, he must come up with a constrained signalling scheme: a (non injective) mapping from goods to signals, that satisfies the constraints of his setting. For example, the auctioneer may be able to communicate only a bounded length message for each good, or he might be legally constrained in how he can advertise the item being sold. Each candidate signaling scheme induces an incomplete-information game among the buyers, and the goal of the auctioneer is to choose the signaling scheme and accompanying auction format that optimizes welfare. In this paper, we use techniques from submodular function maximization and no-regret learning to give algorithms for computing constrained signaling schemes for a variety of constrained signaling problems.
\end{abstract}

 \section{Introduction}

At a cafe in {\it Portlandia}~\cite{Portlandia}, customers about to order the chicken ask the waitress for more details regarding its source. She informs them that the chicken is a heritage breed, woodland raised, and has been fed a diet of sheep's milk, soy, and hazelnuts, and she assures them the chicken is indeed local, free range, and so on.  Ah but, the customers ask, is that USDA Organic, or Oregon Organic, or Portland Organic?

Organic certifications are a means by which a seller can communicate parameters of a product to potential buyers.  The certification system creates a simplified and practical, yet sufficiently expressive, set of signals as the basis for this communication.\footnote{Unfortunately for the waitress in {\it Portlandia}, the language was not sufficiently expressive for her customers.  They decided to drive to the farm to visit the chicken's home.} Signaling is common throughout complex markets.  The USDA classifies meat into a small number of discrete grades.  Graduating high-school students signal their potential to employers and colleges through SAT scores and transcripts.  Targeted advertising sales in both online and offline media describe viewers to ad buyers through a fixed and (comparatively) small set of demographic information.\footnote{See Milgrom~\cite{Milgrom10} and Levin and Milgrom~\cite{LevinM10} for further discussion of signaling in various markets including online advertising, wheat sales, diamond sales, and spectrum auctions.}

We consider a market with one seller, one or more buyers, and multiple potential items.  A {\it signaling scheme} for the market maps each potential item to a signal.  The seller commits to a signaling scheme up front and sells the signal induced by the item through a sale mechanism (auction, posted price, etc.).
This choice of signaling scheme induces a game among buyers: the buyers find themselves looking to buy an item in the mechanism without knowing what item is for sale, but rather only knowing that it is some item that induced a particular signal from a known signaling scheme. In other words, a signaling scheme is a fixed bundling of the goods that the seller commits to up front. In this work, we focus on designing signaling schemes and corresponding auction mechanisms that maximize social welfare. Once the seller commits to a particular signaling scheme, the second price auction maximizes welfare with respect to the agreed upon bundling, and so without loss of generality, the ``auction mechanism'' can be taken to be the simple second price auction. Our goal is to (subject to various constraints) define signaling schemes that maximize the expected welfare in the induced bundle auction.


%
Practical settings impose a variety of constraints on feasible signaling schemes.
In many markets, the {\it amount} of information that can be conveyed is highly constrained.  In online advertising auctions, the sheer volume of sales and diversity of viewers make it impractical to communicate precise details of every viewer.  In addition, various reputational and legal constraints may restrict a seller from announcing certain signals for certain items: e.g., organic products must pass a certification process from the corresponding agency to be sold under that label.
The main question we address is a computational one: {\it how should an auctioneer compute a welfare-maximizing signaling scheme in the presence of exogenously-defined signaling constraints?}
We first observe that welfare maximization in constrained signaling is a convex function maximization problem, and so there is always a deterministic signaling scheme which maximizes expected welfare.  Furthermore, fixing the signaling scheme, the optimal mechanism for welfare maximization is the second-price auction.   These two facts significantly simplify our analysis, as we now only have to search over deterministic signaling schemes.  We also observe that for second-price auctions revenue and welfare are intimately related: the revenue is at most the welfare after excluding an arbitrary player.  This allows us to extend some of our welfare results to the revenue setting.

As our primary signaling constraint, we focus on communication-bounded signaling in which the constraint is the amount of communication which can be used to send the signal. We consider settings in which goods are points in a very high dimensional space $\Omega = \mathbb{R}^d$, where $d$ is much too large to communicate the entire vectors (e.g. $d$ might be exponential in the number of bidders $n$). In these settings, we first consider agents who have {\it geometric valuations}, defined in terms of either the distance or the angle from the realized good to some target good or set of goods. Thus if $d$ represents, for example, the vector of potential features of an item, then the distance or angle represents how close the actual item is to the agent's ideal item (see Section~\ref{sec:geometric} for a more precise motivation). We use techniques from no regret learning and metric embeddings to give highly space efficient signaling schemes that achieve nearly the optimal welfare (compared even to the optimal unconstrained signaling scheme), when values are drawn according to a known prior, and in one instance, even against adversarially chosen goods, and without knowledge of the prior distribution over valuations. We make a novel use of no-regret learning algorithms as compression schemes. Specifically, we view the item for sale as a function labeling each bidder with their valuation for the item. Given a realized item, we learn this function using the multiplicative weights update rule, and then communicate the learned function by sending only the identity of the (very small number of) update vectors, which results in a low-communication scheme.  The communication required by these schemes depends only logarithmic on the dimension $d$ and allows for an infinite set of potential goods.

We next study agents with {\it arbitrary valuations} and show how to completely eliminate the dependence on the dimension for a bounded number of potential goods with a constant loss in the approximation factor.
Our technique here extends to other exogenously imposed constraints, called bipartite signaling, in which we are given a set of feasible signals, and each item is allowed to be mapped to any one of a subset of these signals.\footnote{For example, a signaling scheme for online advertising auctions may consist of all potential tuples of attributes: age, location, ethnicity, gender, and income, say.  A viewer can be matched to any tuple that matches his true characteristics.}
For these results, we make the unrealistic assumption that the seller knows the set of values of the buyers precisely.  Our motivation for studying this setting is two-fold: for one, we hope that it will lead to mechanisms for the unknown valuations case, and indeed some of our results carry through when the amount of uncertainty in the unknown valuations case is bounded in some sense; for another, in large markets with limited numbers of valuation types, an auctioneer can be relatively confident regarding the set of valuations in the market.
When valuations are known to the auctioneer, we are able to reduce welfare maximization in communication-bounded bipartite signaling to submodular function maximization subject to a matroid constraint, thus implying a $(1-1/e)$-approximation for welfare maximization with respect to the optimal communication-bounded signaling scheme. Starting from our approximately welfare-maximizing signaling scheme, and using the connection between revenue and welfare described above as well as careful ``mixing'' of different signals,  we also devise  a constant approximation for revenue maximization in the communication-bounded signaling setting for auctioneers that are constrained to use second-price auctions.  We further show how to extend some of our results to the unknown valuation case when the prior has constant-size support.
We couple our results with a hardness result: even in the communication-bounded signaling setting with known valuations, it is {\sf NP}-hard to approximate either objective to a factor better than $(1-1/e)$ via a reduction from max-cover.

\paragraph{Related Work}
 The study of markets with information asymmetries between sellers and buyers was first introduced by Akerlof~\cite{Akerlof70}. Since then, a rich literature has examined the effects of information revelation, i.e. signaling, in markets and auctions. Most notably, the ``Linkage principle'' of Milgrom and Weber~\cite{MilgromW82} shows that, under some conditions, a seller always increases his revenue by signalling more information about the good for sale. However the linkage principle requires fairly strong assumptions regarding the joint distribution of player valuations for the good. In particular, it rarely holds in settings where players come from different demographics with negatively correlated values for the good. A review of the literature studying the limits of the linkage principle can be found in Emek et al~\cite{EmekFGLT12}. The non-applicability of the linkage principle to the settings we consider goes even beyond these limitations --- even when full transparency is optimal for an auctioneer, the presence of constraints on the amount and nature of information revealed by our signalling schemes introduces intricate tradeoffs in choosing  \emph{which} information to reveal. Quantifying those tradeoffs inevitably requires examination of these settings with an optimization lens, as we do here.

Our work is inspired by, and builds on, recent papers that examine optimal signalling schemes in \emph{unconstrained settings}. Specifically, Emek et al ~\cite{EmekFGLT12} and Miltersen and Sheffet~\cite{MiltersenS12} examine signalling for revenue maximization in a second price auction,  where no constraints are placed on the number or nature of signals. In such settings, full information revelation is optimal for an auctioneer interested in maximizing welfare, and both works show that a revenue-maximizing scheme can be computed efficiently when player valuations are known. Emek et al ~\cite{EmekFGLT12} also obtain partial results for revenue maximization when player valuations are drawn from a Bayesian prior. Our results can be thought of as the extension of these works to settings where social, legal, or practical constraints are placed on the auctioneer's signaling policy.

We also mention some results that are related to ours in the techniques used. In the geometric setting considered in Section \ref{sec:geometric}, in which players value goods $y \in \mathbb{R}^d$ according to \emph{inner product valuations}, our  communication bounded signaling scheme applies the multiplicative weights framework of Arora, Hazan, and Kale \cite{AHK}. We use multiplicative weights in an unusual way, both as a no-regret learning algorithm and as a \emph{compression scheme}. That is, we use both the fact that multiplicative weights can quickly learn a hypothesis $\hat{y}$ that closely approximates a vector $y$ with respect to a fixed number of inner product valuations, and the fact that the hypothesis $\hat{y}$ can be concisely communicated by transmitting only the update operations of multiplicative weights, rather than the vector $\hat{y}$ itself. We are not aware of multiplicative weights being used explicitly as a compression scheme elsewhere, although this is related to the use that no-regret algorithms have recently found in differential privacy \cite{RR10,HR10,GRU12}.

 In Section \ref{sec:bipartite}, we point out that the special case of bipartite signaling where the graph is complete is technically equivalent to the clustering problem considered in \cite{mahdianclustering}. Moreover, the problem of computing a revenue-maximizing bundling of goods  considered in~Ghosh et al.~\cite{GhoshNS07} is both conceptually and technically similar to the optimization problem we face in our revenue-maximizing scheme (Section \ref{sec:revenue}). The idea of ``merging'' signals in our setting is inspired by their algorithm, though the constraint on the number of signals in our setting poses additional technical hurdles. The results of Section \ref{sec:bipartite} also heavily use techniques from the combinatorial auctions literature  (\cite{Feige06, DS06, DNS05}), as well as the submodular function maximization result of Vondrak \cite{vondrak08}.


\section{Preliminaries}
\label{sec:prelim}

We consider a setting in which there is a (possibly infinite) set $\Omega$ of possible items for sale, and a single item $\omega \in \Omega$ is drawn from a distribution $p \in \Delta_\Omega$. There is a set of $n$ players, each of whom is equipped with a \emph{valuation} $v_i : \Omega \to \RR$ mapping items to the real numbers. The valuation profile $(v_1,\ldots,v_n)$ is drawn from a distribution $\D$. We assume that $\Omega$, $p$, and $\D$ are common knowledge, while each player's valuation $v_i$ is private to player $i$.  It will be useful to interpret the possible items $\Omega$ as subsets of $\mathbb{R}^d$, where $d$ is a vector of possible item features. 

We assume that the realization of item $\omega$ is ex-ante unknown to the players, but known to  the auctioneer.  We consider an auctioneer who first observes the drawn item $\omega$, and then announces a string $s$, known as a \emph{signal}. The (possibly randomized) policy by which the auctioneer chooses his signal, which we refer to as a \emph{signaling scheme}, is common knowledge. After players observe the signal $s$, which can be thought of as a random variable correlated with the realization of the item $\omega$, an auction for item $\omega$ is run.

We adopt the perspective of an auctioneer seeking to optimize his choice of signaling scheme and corresponding auction, with the goal of maximizing the expected welfare.  With no constraints on the signal, an auctioneer could generate optimal welfare by announcing the item $\omega$ and running a second-price auction. Our focus is on a constrained auctioneer. The class of constrained signalling problems is defined as follows. 

\begin{definition}
  A \emph{constrained signaling problem}  is a family of instances, each given by:
  \begin{itemize}
  \item  A set $\Omega$ of items, and a distribution $p \in \Delta_{\Omega}$ over these items.
  \item A set $[n]$ of players, where each player $i$ is equipped by a private valuation $v_i: \Omega \to \RRp$. The tuple of valuations $(v_1,\ldots,v_n)$ is drawn from a common prior $\D$.
  \item A set of signals $\S$, and a set $\F \sse \S^\Omega$ of \emph{valid signaling maps.}
  \end{itemize}
\end{definition}

When $\D$ is a trivial prior, i.e. $(v_1,\ldots,v_n)$ are deterministic, we say our signaling problem has \emph{known valuations}; otherwise we say it has \emph{unknown valuations}. A solution to a constrained signaling problem is a \emph{valid  signaling scheme}, defined as a  distribution $x \in \Delta_\F$ over valid signaling maps $\F$, and corresponding auction. We note that given a constrained signaling problem, any valid signaling scheme induces a set of information states, namely the pre-image of the mapping.  Fixing these information states, the well-known second-price auction maximizes welfare.  Therefore, we can assume, without loss of generality, that the auctioneer runs a second-price auction.

When the signaling scheme $x$ is a point distribution, we say the scheme is \emph{deterministic}. Given  $x$ and an item $\omega \in \Omega$, we use $x(\omega)$ to denote the random variable $f(\omega)$ for $f \sim x$. Moreover, given  item $\omega \in \Omega$ and signal $s \in \S$, we abuse notation and use $x(\omega,s)$ to denote the probability that $f(\omega)=s$ for $f \sim x$. Similarly, given a signal $s \in \S$ we use $x(s)$ to denote the probability that $f(\omega)=s$ for $f \sim x$ and $\omega \sim p$; it is easy to see that $x(s) = \sum_{\omega \in \Omega} p_\omega x(\omega,s)$.

A signaling scheme $x$ induces, for each signal $s$, a second-price auction where players have independent private values.
Specifically, it is a dominant strategy for each  player $i$ to bid his value for item $\omega \sim p$ conditioned on signal $s$ --- namely
$ v_i|s,x := \Ex_{f \sim x, \omega \sim \D} [ v_i(\omega) | f(\omega) = s ] = \frac{\sum_{\omega \in \Omega} x(\omega,s) p_\omega v_i(\omega)}{x(s)}$
The winning player for signal $s$ is the player maximizing $v_i | s,x$. Parametrized by the valuation profile $v$, the resulting welfare of the auction, in expectation over all draws of the item $\omega$, is given by
$welfare(x,v) := \sum_{s \in \S} x(s) \max_{i=1}^n v_i | s,x = \sum_{s \in \S}  \max_{i=1}^n  \sum_{\omega \in \Omega} x(\omega,s) p_\omega v_i(\omega).$
Using $\hat{v}_i(\omega) := p_\omega v_i(\omega)$ to denote player $i$'s value for item $\omega$ weighted by $\omega$'s probability, gives
\begin{equation}
 welfare(x,v)=\sum_{s \in \S}  \max_{i=1}^n  \sum_{\omega \in \Omega} x(\omega,s) \hat{v}_i(\omega). \label{eq:welfare}
\end{equation}
In the case of unknown valuations, the expected welfare of the auction over draws of the players' valuations, which we denote by $welfare(x)$, is the expectation of $welfare(x,v)$ over $v \sim \D$.



Because the dimension of the items $d$ is so large, it is not possible to exactly describe them to buyers, and it is instead necessary to employ a \emph{communication bounded signaling scheme} which places a limit on the amount of communication that can be invested in transmitting the signal to the agents.
%
%
Informally, we say that a signalling scheme has $b$-bounded length of the total set of signals it can generate (over all possible items) is of size at most $k=2^b$ -- note that any signal from such a set can be indexed by using at most $b$ bits. We would like to compute the optimal signaling scheme subject to these communication constraints.  We focus on two instances of this problem: we first study {\it geometric valuations} in which we assume a particular form of the valuations of the agents and show how to get very close approximations to the unconstrained optimal welfare.  We then consider arbitrary valuations and show how to get constant-factor approximations to the optimal welfare of constrained schemes (which, in turn, are $(k/n)$-approximations to the unconstrained optimum).  Our techniques for this approximation allow us to handle further {\it bipartite signal constraints} in which signals are labeled by subsets of features and an item can be mapped to a signal only if the item's features are a superset of the features in the signal's label.  

We discuss and derive several basic structural results regarding constrained signaling schemes in Appendix \ref{app:structural}. We here summarize several of the main findings:
\begin{lemma}\label{lem:det_welfare}
  For any constrained signaling problem with unknown valuations, there is valid deterministic signaling scheme which maximizes expected welfare.
\end{lemma}
\begin{lemma}\label{lem:scalingsignals}
  Consider an $n$-player and $m$-item signaling problem with known valuations. For every integer $k$, there is a signaling scheme with $k$ signals and welfare at least a $\frac{k}{\min(n,m)}$ fraction of that of optimal (unconstrained) scheme.
\end{lemma}
\begin{lemma}
  \label{lem:rev_ub}
Fix an arbitrary constrained signaling problem with unknown valuations. Let $i'$ be an arbitrary player. The revenue of the revenue-optimal signaling scheme is at most the welfare of the welfare-optimal signaling scheme for all players other than $i'$.
\end{lemma}

 \section{Constrained signaling with geometric valuations}
\label{sec:geometric}
In this section, we consider signaling schemes in which items $\omega \in \Omega$ correspond to points in $d$ dimensional Euclidean space $\mathbb{R}^d$. In such settings, we consider valuation functions which are also parameterized by points $v_i \in \mathbb{R}^d$. Natural valuation functions then include inner products (i.e. $v_i(\omega) = \langle \omega, v_i \rangle$) and distances (i.e. $v_i(\omega) = ||\omega - v_i||$). We think of $d$ as being very large, and so we will be concerned with \emph{space bounded signaling schemes}, defined to be schemes that can communicate only a bounded number of bits per signal.

\begin{definition}
A finite set of signals $S$ has $b$-bounded length if $\log|S| \leq b$. The set of $b$-bounded length signaling schemes is $\mathcal{F}\subseteq S^\Omega$ such that for each $f \in F$: the set of all possible signals generated by $f$ (i.e. $\cup_{\omega \in \Omega}f(\omega))$ has $b$ bounded length. We call such an $f \in F$ a $b$-bounded length signaling scheme.
\end{definition}
\begin{remark}
Note that any signal in a $b$-bounded length signaling scheme can be transmitted using at most $b$ bits.
\end{remark}

Note that in any known valuation setting, there is always a $\log n$-bounded length signalling scheme that achieves full welfare: the scheme simply partitions the items into $n$ sets in which the $i$'th set consists of all items that the $i$'th bidder likes better than anyone else.  This scheme, in effect, names the bidder who has the highest valuation for the realized good. Our main focus in this section will therefore be on achieving bounded length signalling schemes in the more demanding \emph{unknown valuation} setting, in which bidder valuations are either drawn from a known prior, or selected by an adversary.

\subsection{Inner Product Valuations}
Let $\Omega = \{\omega \in \mathbb{R}_+^d : ||\omega||_1 = 1\}$ be the set of $d$ dimensional non-negative real vectors with $\ell_1$ weight 1. These can be thought of as distributions over $d$ ``features'' which describe the product $\omega$. Individuals $i$ have valuation functions $v_i:\Omega\rightarrow \mathbb{R}$ parameterized by (abusing notation) a vector $v_i \in \mathbb{R}^d$ with $||v_i||_\infty \leq 1$. Bidder $i$'s valuation for good $\omega$ is defined to be $v_i(\omega) \equiv \langle v_i, \omega \rangle$. Intuitively, each component $j$ of $v_i$ represents agent $i$'s affinity for feature $j$, which can range in $[-1,1]$. $v_i(\omega)$ is agent $i$'s average affinity for the distribution over features represented by good $\omega$. Here $d$ can be exponentially large, and so such valuation functions are extremely general. We wish to design good $b$-bounded length signaling schemes for inner product valuations, for $b$ as small as possible, when bidder valuations are unknown but drawn independently from a known prior.

\subsubsection{Multiplicative Weights Signaling}
In this section, we use the versatile multiplicative weights framework of \cite{AHK} to give a communication bounded signaling scheme. The idea is the following: we treat the realized good $\omega \in \Omega$ as a \emph{function} defined over the domain of valuation vectors $v_i$. The function $\omega$ labels each possible valuation vector $v_i$ with the real number $\langle v_i, \omega \rangle$. Given a realized good, we then attempt to learn this function over the prior distribution $\mathcal{P}$ from which the valuations are drawn. If we are able to ``correctly'' classify new examples drawn from the distribution with probability at least $1-\delta/n$, then with probability at least $1-\delta$, we are able to correctly classify all of the actual bidder valuation functions drawn i.i.d. from $\mathcal{P}$. If we know the prior $\mathcal{P}$ and have access to i.i.d. samples from it, then we are in a PAC-like setting, and are able to inherit very strong PAC-like bounds: learning algorithms with guarantees that hold for \emph{any} distribution $\mathcal{P}$. If we are able to learn our classifier in a way that can be concisely communicated, then we also have a good bounded length signalling scheme.

For ease of exposition, we first consider the known valuation setting in which the signaling scheme can be parameterized by the actual valuations $v_1,\ldots,v_n$ of the bidders\footnote{As mentioned earlier, it is trivial to derive a bounded-length signalling scheme in the known valuation setting without the machinery of multiplicative weights. We introduce this machinery in the known valuation setting and then show how it easily extends to the unknown valuation setting.}. When a good $\omega \in \Omega$ arrives, we use multiplicative weights to learn an $\epsilon$-approximate representation $\hat{\omega}$ with respect to the bidder valuation functions $v_1,\ldots,v_n$ -- that is, a representation $\hat{\omega}$ such that for all $i$, $|v_i(\omega) - v_i(\hat{\omega})| \leq \epsilon$. We can do this by updating multiplicative weights at most $O(\log d/\epsilon^2)$ times using the valuation functions themselves as loss functions. $\hat{\omega}$ would itself therefore make a terrific signal -- it would approximately represent every bidder's valuation for the good $\omega$. However, $\hat{\omega}$ is also a vector in $\mathbb{R}^d$, and so it is not clear why we should be able to communicate it in a space bounded way. The key insight is that it is not necessary to communicate $\hat{\omega}$ directly, but merely communicate which collection of valuation functions were used to update multiplicative weights when learning $\hat{\omega}$ -- using this information, each bidder can reconstruct $\hat{\omega}$ for themselves. (Of course, this ``reconstruction'' can be automated, so the bidders can still see a natural signal). Since each valuation function can be indexed with only $O(\log n)$ bits, and there are only $O(\log d/\epsilon^2)$ updates in total, this gives a $O(\log n \log d/\epsilon^2)$-space bounded signaling scheme that approximates the optimal welfare within an additive loss of $\epsilon$. We then extend this to the case in which the valuation functions $v_i$ are not known, but instead drawn i.i.d. from a known prior $\mathcal{P}$. This extension involves parameterizing the same multiplicative weights signaling scheme with $m$ i.i.d. samples from $\mathcal{P}$. Here $m$ corresponds to the \emph{sample complexity} of the corresponding learning problem on linear valuations\footnote{The \emph{sample complexity} of a learning problem is, informally, the number of samples that need to be drawn from a distribution $\mathcal{P}$, such that if we learn a hypothesis that is consistent on the sampled points, then with high probability, the hypothesis is consistent on \emph{new} points drawn from the same distribution.} and the space used by this signaling scheme depends on $\log m$ (since we must index updates from this set of $m$ vectors). We now define our signaling scheme formally.

We define a signaling scheme $f_{\textrm{MW},\epsilon, z_1,\ldots,z_m}$ parameterized by a no regret algorithm (in this case the multiplicative weights algorithm $\textrm{MW}$), an accuracy parameter $\epsilon$, and $m$ vectors $z_1,\ldots,z_m \in \mathbb{R}^d$. The parameters of the signaling scheme will be public knowledge, and it will be used to generate $b$-bounded length signals as follows:

\begin{algorithm}
\caption{Algorithm for computing the signal $f_{\textrm{MW}, \epsilon, z_1,\ldots,z_m}$}
\begin{algorithmic}[1]
\INPUT An instance $\omega \in \Omega$.
\OUTPUT A bounded length signal $s$.
\STATE Initialize $\hat{\omega}^1 \in \mathbb{R}^d$ such that $\hat{\omega}^1_j = 1/d$ for all $j \in [d]$.
\STATE Initialize $T \leftarrow 1$.
\WHILE{there exists an index $i$ such that $|\langle z_i, \omega \rangle - \langle z_i, \hat{\omega}^T \rangle| \geq \epsilon/2$}
  \STATE Let $\textrm{Update}_T \leftarrow i$, $\textrm{Sign}_T \leftarrow \mathrm{sign}(\langle z_i, \hat{\omega}^T \rangle - \langle z_i, \omega \rangle)$,
  \STATE For all $j$ let $\hat{\omega}^{T+1}_j \leftarrow \hat{\omega}^T_j \cdot \left(1-\mathrm{Sign}_T\cdot\frac{\epsilon}{4}\cdot z_{i,j}\right)$.
  \STATE Normalize $\hat{\omega}^{T+1}$ such that $||\hat{\omega}^{T+1}||_1 = 1$. Let $T \leftarrow T+1$
\ENDWHILE
\STATE Let $s = ((\textrm{Update}_1,\textrm{Sign}_1), \ldots, (\textrm{Update}_T,\textrm{Sign}_T))$
\end{algorithmic}
\end{algorithm}

\begin{theorem}
\label{thm:boundedupdates}
For any $\epsilon$, any vectors $z_1,\ldots,z_m$ such that for all $i$, $||z_i||_\infty \leq 1$, and any $\omega \in \Omega$, $f_{\textrm{MW}, \epsilon, z_1,\ldots,z_m}$ runs for $T \leq \frac{16\log d}{\epsilon^2}$ rounds.
\end{theorem}
\begin{proof}
$f_{\textrm{MW}, \epsilon, z_1,\ldots,z_m}(\omega)$ runs an instantiation of the Multiplicative Weights Framework for Arora, Hazan, and Kale \cite{AHK} for $T$ rounds, with update parameter $\epsilon/4$, and loss vectors at each round $t$ defined to be $\ell^t = \mathrm{Sign}_t\cdot z_{\textrm{Update}_t}$. By the regret bound of Multiplicative Weights (see e.g. \cite{AHK} Corollary 2.2), we have for all $x \in \mathbb{R}^d_+$ with $||x||_1 = 1$:
$$\sum_{t=1}^T\left(\langle \ell^t, \hat{\omega}^t \rangle - \langle \ell^t, x \rangle \right) \leq \frac{\epsilon}{4}\sum_{t=1}^T\langle |\ell^t|, x \rangle + \frac{4 \ln d}{\epsilon}$$
where $|\ell^t|$ denotes a coordinate-wise absolute value. Note that this corresponds to the standard ``no regret bound'' that readers may be more familiar with: here the ``experts'' correspond to the $d$ standard basis vectors, $x$ corresponds to a distribution over experts, and the above bound simply states that multiplicative weights achieves diminishing regret with respect to the best expert (and therefore with respect to any distribution over experts). Taking $x = \omega$, we note that the loss vectors $\ell^t = \mathrm{Sign}_t\cdot z_{\textrm{Update}_t}$ have been constructed such that at every round $t$, $\langle \ell^t, \hat{\omega}^t \rangle - \langle \ell^t, \omega \rangle \geq \epsilon/2$. Moreover, since for all $i$ $||z_i||_{\infty} \leq 1$, by definition for each coordinate $j$, $|\ell^t_j| \leq 1$. Therefore, for each $t$, $\langle |\ell^t|, \omega \rangle \leq 1$ since $||\omega||_1 = 1$ . Therefore, the above bound becomes:
$$\frac{T\epsilon}{2} \leq \frac{T\epsilon}{4} + \frac{4 \ln d}{\epsilon}$$
Solving for $T$, we find that it must be that $T \leq \frac{16\log d}{\epsilon^2}$ as desired.
\end{proof}
We now make two observations. The first is that $f_{\textrm{MW}, \epsilon, z_1,\ldots,z_m}(\omega)$ produces bounded length signals:
\begin{corollary}
For any $\epsilon$, any vectors $z_1,\ldots,z_m$ such that for all $i$, $||z_i||_\infty \leq 1$, and any $\omega \in \Omega$, $f_{\textrm{MW}, \epsilon, z_1,\ldots,z_m}$ is a $b$-bounded length signaling scheme for $b = \left(\frac{16\log d(\log m+1)}{\epsilon^2}\right)$.
\end{corollary}
\begin{proof}
For $t = 1,\ldots,T$, $\textrm{Update}_t$ can be communicated with $\log m$ bits and $\textrm{Sign}_t$ can be communicated by 1 bit. By Theorem \ref{thm:boundedupdates}, for all $\omega \in \Omega$, $T \leq \frac{16\log d}{\epsilon^2}$
\end{proof}
The next is that the signal $s = ((\textrm{Update}_1,\textrm{Sign}_1), \ldots, (\textrm{Update}_T,\textrm{Sign}_T))$ is sufficient for each agent $i$ to reconstruct $\hat{\omega}^{T+1}$:
\begin{observation}
 $\hat{\omega}^{T+1} \equiv \hat{\omega}^{T+1}(s)$ is a function only of the signal $s = ((\textrm{Update}_1,\textrm{Sign}_1), \ldots, (\textrm{Update}_T,\textrm{Sign}_T))$.
\end{observation}

Given this observation, it is helpful to think about the signal ``really being'' the vector $\hat{\omega}^{T+1}$ that well approximates $\langle \omega, z_i \rangle$ for all $i$. The signal $s$ is just a concise way of transmitting this vector.

First, we show that when we have known valuations, the multiplicative weights signaling scheme is competitive with the optimal $b$-bounded length signaling scheme. In fact, we show more -- that the multiplicative weights signalling scheme is competitive with the optimal \emph{unconstrained} signalling scheme, even pointwise. This proof will be a template for the more interesting unknown valuations case.

\begin{theorem}
\label{thm:welfareMW1}
Let $\textrm{OPT}_U = \sum_{\omega \in \Omega}\Pr_{j \sim p}[j = \omega]\cdot\max_{i \in [n]}v_i(\omega)$ denote the optimal social welfare in the \emph{unconstrained}  setting.
In the known valuation setting, the welfare obtained by the multiplicative weights signaling scheme given vectors $(z_1,\ldots,z_n) \equiv (v_1,\ldots,v_n)$ is:
$$\textrm{welfare}(f_{\textrm{MW}, \epsilon, v_1,\ldots,v_n}, v) \geq \textrm{OPT}_{U} - \epsilon$$
In particular, it is within $\epsilon$ of the optimal welfare obtained by \emph{any} bounded length signaling scheme.
\end{theorem}
\begin{proof}
Fix any $\omega \in \Omega$, consider $s = s(\omega) = f_{\textrm{MW}, \epsilon, v_1,\ldots,v_n}(\omega)$, and let $\hat{\omega}^{T+1} \equiv \hat{\omega}^{T+1}(s)$. For any bidder $i$, let $S = \{\omega : |\langle \omega, v_i \rangle - \langle \hat{\omega}^{T+1}(s), v_i \rangle| \leq \epsilon/2\}$ be the set of goods whose value to bidder $i$ differs by less than $\epsilon$ from the value of good $\hat{\omega}^{T+1}$. Note that by the construction of $\hat{\omega}^{T+1}$, it must be the case that $\Pr_{\omega \sim p}[\omega \not \in S | s] = 0$ (because otherwise the multiplicative weights signaling scheme would not have halted). Therefore, we have $\Pr_{\omega \sim p}[\omega \in S | s] = 1$ and we can calculate:
\begin{align*}
v_i|s(\omega) &= \mathbb{E}_{j \sim p}[v_i(j)|f(j) = s]\\
 &= \sum_{j \in \Omega}v_i(j)\cdot \Pr_{j \sim p}[j | f(j) = s]\\
&\geq  \sum_{j \in S}v_i(j)\cdot \Pr_{j \sim p}[j | f(j) = s] \\
&\geq \min_{j \in S}v_i(j)\cdot \Pr_{j \sim p}[j \in S | f(j) = s] \\
&\geq ( v_i(\omega)-\epsilon)\cdot 1
\end{align*}

Finally, we can lower bound the expected welfare of the multiplicative weights signaling scheme, and compare it to the unconstrained optimal welfare. For all $v$:
{\scriptsize
\begin{align*}
\textrm{welfare}(f_{\textrm{MW}, \epsilon, v_1,\ldots,v_n}, v) &= \sum_{\omega \in \Omega}\Pr_{j \sim p}[j = \omega]\cdot \max_{i\in [n]}v_i|s(\omega) \\
 &\geq \sum_{\omega \in \Omega}\Pr_{j \sim p}[j = \omega]\cdot \max_{i \in [n]}(v_i(\omega) - \epsilon) \\
&= \textrm{OPT}_{U} - \epsilon.
\end{align*}
}
 \end{proof}

We now adapt our signaling scheme slightly, and show that it works not just in the known valuation setting, but also in the Bayesian setting when there is a prior $\mathcal{P}$ from which the valuations $v_i$ of the agents are drawn i.i.d. Consider the slightly modified family of signaling schemes $f'$, still defined in terms of an accuracy parameter $\epsilon$ and $m$ vectors $z_1,\ldots,z_m$. The difference between $f'$ and $f$ is only that $f'$ updates the multiplicative weights hypothesis at most $1$ time for every vector $z_i$, and halts once it finds a hypothesis which does not induce an update for a sufficiently large sequence of the $z_i$ vectors. The idea is simple: The $z_i$ vectors will be drawn i.i.d. from the prior $\mathcal{P}$. At any given time, either the current multiplicative weights hypothesis will have low error over new examples drawn from $\mathcal{P}$, in which case we can halt, or it will have high error, which means that it is likely to induce an update on one of the next few $z_i$'s it iterates through. Given that we have a bound on the total number of updates that it can perform, it is not hard to see that it must quickly find a hypothesis vector $\hat{\omega}$ that well approximates $\omega$ on a large measure of examples $z_i$ drawn from $\mathcal{P}$.

\begin{algorithm}
\caption{Algorithm for computing the signal $f'_{\textrm{MW}, \epsilon, z_1,\ldots,z_m}$}
\begin{algorithmic}[1]
\INPUT An instance $\omega \in \Omega$.
\OUTPUT A bounded length signal $s$.
\STATE Initialize $\hat{\omega}^1 \in \mathbb{R}^d$ such that $\hat{\omega}^1_j = 1/d$ for all $j \in [d]$.
\STATE Initialize $T \leftarrow 1$.   \texttt{// Indexes Updates}
\STATE Initialize $c \leftarrow 0$.   \texttt{// Counts rounds between updates}
\STATE Let $r \leftarrow \frac{2n\left(\log\left(\frac{16\log d}{\epsilon^2}\right) + \log\frac{2n}{\delta}\right)}{\delta}$ \texttt{// Threshold to halt}.
\FOR{$i = 1$ to $m$}
  \IF{$|\langle z_i, \omega \rangle - \langle z_i, \hat{\omega}^T \rangle| \geq \epsilon/2$}
    \STATE Let $\textrm{Update}_T \leftarrow i$, $\textrm{Sign}_T \leftarrow \mathrm{sign}(\langle z_i, \hat{\omega}^T \rangle - \langle z_i, \omega \rangle)$,
    \STATE For all $j$ let $\hat{\omega}^{T+1}_j \leftarrow \hat{\omega}^T_j \cdot \left(1-\mathrm{Sign}_T\cdot\frac{\epsilon}{2}\cdot z_{i,j}\right)$.
    \STATE Normalize $\hat{\omega}^{T+1}$ such that $||\hat{\omega}^{T+1}||_1 = 1$. Let $T \leftarrow T+1$. Let $c \leftarrow 0$.
  \ELSE
     \STATE $c \leftarrow c + 1$
     \IF{$c \geq r$}
       \STATE Output $$s = ((\textrm{Update}_1,\textrm{Sign}_1), \ldots, (\textrm{Update}_T,\textrm{Sign}_T))$$ and HALT.
     \ENDIF
  \ENDIF
  \STATE Output $$s = ((\textrm{Update}_1,\textrm{Sign}_1), \ldots, (\textrm{Update}_T,\textrm{Sign}_T))$$
\ENDFOR
\end{algorithmic}
\end{algorithm}

Because we continue to have that $\hat{\omega}^T$ is updated only on vectors $z_i$ such that $|\langle z_i, \omega \rangle - \langle z_i, \hat{\omega}^T \rangle| \geq \epsilon/2$, our bound on the number of updates is identical as it was for signaling scheme $f$, and so we have an identical corollary:

\begin{corollary}
For any $\epsilon$, any vectors $z_1,\ldots,z_m$ such that for all $i$, $||z_i||_\infty \leq 1$, and any $\omega \in \Omega$, $f'_{\textrm{MW}, \epsilon, z_1,\ldots,z_m}$ is a $b$-bounded length signaling scheme for $b = \left(\frac{16\log d(\log m+1)}{\epsilon^2}\right)$.
\end{corollary}

We now argue that if $f$ is parameterized with $m$ vectors $z_i$ drawn i.i.d. from $\mathcal{P}$ (for sufficiently large $m$), then in fact with high probability, $f'_{\textrm{MW}, \epsilon, z_1,\ldots,z_m}$ is a competitive signaling scheme for the actual agent valuations $v_1,\ldots,v_n$, whenever $v_1,\ldots,v_n$ are also drawn i.i.d. from $\mathcal{P}$.

\begin{theorem}
\label{thm:bayesianconsistency}
Let
$$m = \frac{2n\left(\frac{16\log d}{\epsilon^2}\right)\left(\log\left(\frac{16\log d}{\epsilon^2}\right) + \log\frac{2n}{\delta}\right)}{\delta}$$$$= \tilde{O}\left(\frac{n\log d}{\delta\epsilon^2}\right).$$
Fix any $\omega \in \Omega$, and let $z_1,\ldots,z_m$ and $v_1,\ldots,v_n$ be i.i.d. draws from $\mathcal{P}$. Let $s =  f'_{\textrm{MW}, \epsilon, z_1,\ldots,z_m}(\omega)$, and define $S(v,s) = \{\omega : |\langle \omega, v \rangle - \langle \hat{\omega}^{T+1}(s), v \rangle| \leq \epsilon/2\}$. Then we have:
$$\Pr_{z_1,\ldots,z_m,v \sim \mathcal{P}}[\omega \not\in S(v,s)] \leq \frac{\delta}{n}$$
\end{theorem}
\begin{proof}
We will show that for any $\omega \in \Omega$, with probability $1-\delta/2n$ over the choice of $z_1,\ldots,z_m$, $s$ is such that: $\Pr_{v \sim \mathcal{P}}[\omega \not\in S(v,s)] \leq \delta/2n$, which is enough to prove our claim. Let $U = \frac{16\log d}{\epsilon^2}$ be the maximum number of update rounds that $f'_{\textrm{MW}, \epsilon, z_1,\ldots,z_m}$ can ever conduct, as bounded by Theorem \ref{thm:boundedupdates}. Since $f'_{\textrm{MW}, \epsilon, z_1,\ldots,z_m}$ considers running an update on each $z_i$ in sequence but never conducts more than $U$ updates, there must be some consecutive sequence $z_i,\ldots,z_j$ of length at least $(j-i+1) \geq \frac{m}{U} = r$ on which no updates are performed. By design, the algorithm outputs a signal after such a sequence occurs. We will show that if $\Pr_{v \sim \mathcal{D}}[\omega \not\in S(v,s)] \geq \delta/2n$, then the probability of this event occurring is at most $\delta/2n$. Note that if $\Pr_{v \sim \mathcal{P}}[\omega \not\in S(v,s)] \geq \delta/2n$, since each $z_i$ is independently sampled from $\mathcal{P}$, the probability of an update occurring at round $i$ is at least $\delta/2n$. Therefore, the probability of there existing such a long sequence between updates is at most:
$$(U)(1-\frac{\delta}{2n})^{m/U} \leq (U)e^{-m\delta/2Un}$$
Setting $m$ as in the theorem statement makes this probability at most $\delta/2n$ as desired.
\end{proof}
A simple corollary of Theorem \ref{thm:bayesianconsistency} is an analogous welfare guarantee for the multiplicative weights signaling scheme in the Bayesian setting:
\begin{theorem}
\label{thm:welfareMW2}
Let $\textrm{OPT}_U = \sum_{\omega \in \Omega}\Pr_{j \sim p}[j = \omega]\cdot\max_{i \in [n]}v_i(\omega)$ denote the optimal social welfare in the \emph{unconstrained} setting with known valuations.
In the Bayesian setting in which each $v_i \sim \mathcal{P}$ is sampled i.i.d. from a known prior $\mathcal{P}$, the welfare obtained by the multiplicative weights signaling scheme given vectors $(z_1,\ldots,z_m)$ sampled independently from the prior distribution $\mathcal{P}$ is:
$$\textrm{welfare}(f'_{\textrm{MW}, \epsilon, z_1,\ldots,z_m}, v) \geq \textrm{OPT}_{U} - \epsilon - \delta$$
where $m = \tilde{O}\left(\frac{n\log d}{\delta\epsilon^2}\right)$.
In particular, it is a $16\log d(\log m+1)/\epsilon^2$ bounded length signaling scheme that obtains welfare within $\epsilon+\delta$ of the optimal welfare obtained by \emph{any} bounded length signaling scheme.
\end{theorem}

\begin{proof}
Independently for every $\omega \in \Omega$, consider $s = s(\omega) = f_{\textrm{MW}, \epsilon, z_1,\ldots,z_n}(\omega)$, and let $\hat{\omega}^{T+1} \equiv \hat{\omega}^{T+1}(s)$. For any bidder $i$, let $S_i = \{\omega : |\langle \omega, v_i \rangle - \langle \hat{\omega}^{T+1}(s), v_i \rangle| \leq \epsilon/2\}$ be the set of goods whose value to bidder $i$ differs by less than $\epsilon$ from the value of good $\hat{\omega}^{T+1}$. By Theorem \ref{thm:bayesianconsistency}, $\Pr_{j \sim p}[j \in S_i | s] \geq 1-\delta/n$ and so by a union bound $\Pr_{j \sim p}[\forall i, j \in S_i | s] \geq 1-\delta$ and we can calculate for all $i$:
\begin{eqnarray*}
v_i|s(\omega) &=& \mathbb{E}_{j \sim p}[v_i(j)|f(j) = s] \\
&=& \sum_{j \in \Omega}v_i(j)\cdot \Pr_{p}[j | f(j) = s] \\
&\geq&  \sum_{j \in S}v_i(j)\cdot \Pr_{p}[j | f(j) = s] \\
&\geq& \min_{j \in S}v_i(j)\cdot \Pr_{p}[j \in S | f(j) = s] \\
&\geq& (v_i(\omega) - \epsilon)\cdot (1-\delta) \\
&\geq& v_i(\omega) - \epsilon - \delta
\end{eqnarray*}
where the last inequality follows from the fact that $v_i(\omega) \leq 1$.
Finally, we can lower bound the expected welfare of the multiplicative weights signaling scheme, and compare it to the unconstrained optimal welfare. For all $v$:
{\scriptsize
\begin{eqnarray*}
\textrm{welfare}(f'_{\textrm{MW}, \epsilon, z_1,\ldots,z_m}, v) &=& \sum_{\omega \in \Omega}\Pr_{j \sim p}[j = \omega]\cdot \max_{i\in [n]}v_i|s(\omega) \\
&\geq& \sum_{\omega \in \Omega}\Pr_{j \sim p}[j = \omega]\cdot \max_{i \in [n]}(v_i(\omega) - \epsilon - \delta) \\
&=& \textrm{OPT}_{U} - \epsilon - \delta
\end{eqnarray*}
}
\end{proof}

\begin{remark}
We make two remarks about our multiplicative weights signalling scheme: first, the guarantees hold \emph{pointwise} -- that is, even conditioned on the realization of the good. Second, the technique easily extends to bidders whose values are drawn independently but not identically. If each bidder has a valuation drawn from a unique distribution, we simply need to run the algorithm with $m$ samples drawn from \emph{each} distribution. This increases the number of samples needed by a factor of $n$, but only increases the communication needed by the signalling scheme by an additive $\log n$.
\end{remark}
\subsection{Subspace Valuations}
In this section, we use a version of the Johnson Lindenstrauss lemma, for matrices implicitly defined by limited independence families of hash functions, to give bounded space signaling schemes for ``subspace valuations''. Subspace valuations are defined over a point set $\Omega$ of unit vectors in Euclidean space, and can be seen as a generalization of the ``inner product valuations'' considered in the previous section\footnote{The class of valuations may be seen as a generalization of inner product valuations, but the signaling scheme here works in a different range of parameters as the multiplicative weights signaling scheme. Specifically, for the multiplicative weights signaling scheme, we assumed that $\Omega$ consisted of the set of unit vectors in $\ell_1$ space, and that valuation functions were defined by unit vectors in $\ell_\infty$ space. In this section, both points $\omega \in \Omega$, and the vectors which parameterize agent valuation functions are unit vectors in $\ell_2$ space.}. An agent may specify up to $k$ points in $\Omega$, which indicates that he is equally happy with any linear combination of these $k$ points (i.e. his $k$ points define a subspace). His value for a good is defined to be its distance to this subspace. The Johnson Lindenstrauss lemma lets us take a ``projection'' of the point into a lower dimensional space, in such a way that with high probability, each agent is able to estimate its value with high probability. We use the fact that limited independence JL matrices can be concisely represented -- this allows us to construct a new projection matrix for every good, and include its description as part of our signal. This allows us to give a strong Bayesian guarantee -- our algorithm achieves close to optimal welfare against any (arbitrarily correlated) prior value distribution, even without knowledge of the distribution. Indeed, our algorithm works even against \emph{adversarially selected} goods $\omega \in \Omega$, that need not be drawn from any distribution!

Let  $\Omega = \{\omega \in \mathbb{R}^d : ||\omega||_2 = 1\}$ be the set of $d$ dimensional unit vectors in Euclidean space. Individuals $i$ have valuation functions parameterized by $\ell_i$ orthogonal unit vectors $z^i_1,\ldots,z^i_{\ell_i}$ for some $\ell_i \leq k$. These vectors define a subspace $S_i \equiv \mathrm{span}(z^i_1,\ldots,z^i_{\ell_i})$, and the value that agent $i$ has for a good $\omega$ is the distance between $\omega$ and agent $i$'s subspace $S_i$:
$$v_i(\omega) \equiv 1-d(\omega, S_i) = 1-\min_{x \in S_i}||x-\omega||_2$$
For each agent $i$, we can think of the vectors $z^i_1,\ldots,z^i_{\ell_i}$ as specifying up to $k$ ``ideal'' goods, and that agent $i$ is equally happy with any linear combination of his ideals. His valuation for a good $\omega$ drops off with the distance from $\omega$ to his set of ideal goods. Note that it is not necessary that agent $i$ actually specify \emph{orthogonal goods} -- if his ideal goods $z^i_1,\ldots,z^i_{\ell_i}$ are not orthogonal, we can simply orthonormalize them (using, say, the Gram Schmidt algorithm), which does not alter the subspace that they define.

We recall that for any $\omega$, $\arg\min_{x \in S_i}||x-\omega||_2 = \sum_{j=1}^{\ell_i}\langle \omega, z^i_j \rangle \cdot z_j$, and so we can write: $d(\omega, S_i) = \sqrt{1-\sum_{j=1}^{\ell_i}\langle x, z^i_j \rangle^2}$.

To achieve a low space signaling scheme for subspace valuations, we make use of the Johnson-Lindenstrauss lemma.

We will use (a corollary of) a limited-independence version of the Johnson-Lindenstrauss lemma presented in \cite{KN11}, first proven by \cite{Ach01, CW09}. This version of the lemma holds even for concisely represented projection matrices, which allows us to communicate the projection matrix itself as part of our signal. The advantage of doing this is that our algorithm will get strong utility guarantees even in the \emph{prior free} setting, when valuation functions can be drawn from a worst-case (arbitrarily correlated) prior on distributions, and goods can be selected by an adversary, rather than being drawn from any distribution.

\begin{corollary}
[\cite{Ach01,CW09,KN11}]
\label{cor:JL}
For $d > 0$ an integer and any $0 < \epsilon, \delta < 1/2$, let $A$ be a $T\times d$ random matrix with $\pm 1/\sqrt{T}$ entries that are $r$-wise independent for $T \geq 4\cdot 64^2\epsilon^{-2}\log(1/\delta)$ and $r \geq 2\log(1/\delta)$. Then for any $x, \omega \in \mathbb{R}^d$:
$$\Pr_A[|\langle(Ax),(Ay)\rangle - \langle x, \omega\rangle | \geq \frac{\epsilon}{2} (||x||_2^2+||\omega||_2^2)] \leq 2\delta$$
\end{corollary}

\begin{algorithm}
\caption{Algorithm for computing the signal $f_{\textrm{JL}, k, \epsilon}$}
\begin{algorithmic}[1]
\INPUT An instance $\omega \in \Omega$.
\OUTPUT A bounded length signal $s$.
\STATE Generate a $T\times d$ random $\pm 1/\sqrt{T}$-valued matrix $A$ with $r$-wise independent entries for $r = 2\log(3nk/\epsilon)$ and $T = \frac{131072k^2\log(3n/\epsilon)}{\epsilon^4}$.
\STATE Let $\hat{\omega} = Ay$ and let $\hat{\omega}' = \hat{\omega}$ discretized to $\log(3d/\epsilon)$ bits of precision.
\STATE Let $s = (A, \hat{\omega}')$.
\end{algorithmic}
\end{algorithm}

\begin{remark}
The matrix $A$ in our algorithm will be implicitly represented by a hash function mapping coordinates of the matrix $A$ to their values. There are various ways to select a hash function from a family of $r$-wise independent hash functions mapping $[T\times d] \rightarrow \{0,1\}$. The simplest, and one that suffices for our purposes, is to select the smallest integer $s$ such that $2^s \geq T \times d$, and then to let $g$ be a random degree $r$ polynomial in the finite field $\mathbb{GF}[2^s]$. Selecting and representing such a function takes time and space $O(r\cdot s) = O(r(\log d + \log T))$. $g$ is then an unbiased $r$-wise independent hash function mapping $\mathbb{GF}[2^s] \rightarrow \mathbb{GF}[2^s]$. Taking only the last output bit gives an unbiased $r$-wise independent hash function mapping $[s\times d]$ to  $\{0,1\}$, as desired.
\end{remark}
We first observe that $f_{\textrm{JL}, k, \epsilon}$ generates bounded length signals.
\begin{observation}
$f_{\textrm{JL}, k, \epsilon}$ is an $\eta$-bounded length signaling scheme for:
$$\eta = O\left(\frac{k^2\log(n/\epsilon)\log(d/\epsilon)}{\epsilon^4}(\log d + \log (k\log (n/\epsilon)))\right)$$
\end{observation}
\begin{proof}
This follows directly from our choice of $r$ and $T$, the fact that $A$ can be represented using $O(r(\log d + \log T))$ bits, and the fact that we discretize each coordinate of $\hat{\omega}'$ to $\log(3d/\epsilon)$ bits of precision.
\end{proof}

We now show that $f_{\textrm{JL}, k, \epsilon}$ is welfare competitive with the optimal unconstrained signaling scheme.

\begin{theorem}
\label{thm:JLwelfare}
Let $\textrm{OPT}_U = \sum_{\omega \in \Omega}\Pr_{j \sim p}[j = \omega]\cdot\max_{i \in [n]}v_i(\omega)$ denote the optimal social welfare in the \emph{unconstrained} setting.
For every distribution $\mathcal{D}$ over subspace valuation functions and every distribution $p$ over goods $\omega \in \Omega$,
$$\textrm{welfare}(f_{\textrm{JL}, k,\epsilon},v) \geq \textrm{OPT}_{U} - \epsilon$$
In fact, for every $v$, this guarantee holds pointwise for goods $\omega \in \Omega$ even if adversarialy chosen. For all $\omega \in \Omega$, for $s =f_{\textrm{JL}, k,\epsilon}(\omega)$ :
$$\max_{i \in [n]}\sum_{j \in \Omega}x(j,s)p_jv_i(j) \geq \max_{i \in [n]} v_i(\omega) - \epsilon$$
\end{theorem}

\begin{proof}
It suffices to prove the second, stronger claim. Fix a good $\omega \in \Omega$ and let $f_{\textrm{JL}, k, \epsilon}(\omega) = s \equiv (A, \hat{\omega}')$. For each bidder $i$, define $S_i(A, \hat{\omega}, \delta) = \{x \in \Omega : |\sqrt{1-\sum_{j=1}^{\ell_i}\langle x, z^i_j \rangle^2} - \sqrt{1-\sum_{j=1}^{\ell_i}\langle \hat{\omega}, Az^i_j \rangle^2}| \leq \delta\}$ By our choice of $T$, we have: $\Pr[\omega \not \in S_i(A, \hat{\omega}, \epsilon/3)] \leq \epsilon/3n$, where the probability is taken over the choice of projection matrix $A$ Moreover, by our choice of discretization, we have $S_i(A, \hat{\omega}, \epsilon/3) \subset S_i(A, \hat{\omega}', 2\epsilon/3)$. Therefore by a union bound: $\Pr[\exists i : \omega \not \in S_i(A, \hat{\omega}', 2\epsilon/3)] \leq \epsilon/3$. Therefore, for all $i$:
{\scriptsize
\begin{eqnarray*}
\sum_{j \in \Omega}x(j,s) p_j v_i(j) &\geq& \sum_{j \in S_i(A, \hat{\omega}', 2\epsilon/3)}x(j,s) p_j v_i(j) \\
&\geq& \min_{j \in S_i(A, \hat{\omega}', 2\epsilon/3)} v_i(j)\sum_{j \in S_i(A, \hat{\omega}', 2\epsilon/3)}x(j,s) p_j \\
&\geq& (v_i(\omega) - 2\epsilon/3)\Pr[\omega \in S_i(A, \hat{\omega}', 2\epsilon/3)] \\
&\geq& v_i(\omega) - \epsilon
\end{eqnarray*}
}
\end{proof}


\section{Constrained  Signaling with Arbitrary Valuations}
\label{sec:bipartite}
In this section we present both positive and negative results for welfare  and revenue maximization for signaling with bipartite and communication constraints, without making assumptions on the structure of the valuations. In the known valuations case,  we show the existence of a $1-1/e$ approximation algorithm for welfare maximization in bipartite signaling, and a constant approximation algorithm for revenue maximization in communication-constrained signaling. Our results extend to the unknown valuations case when the Bayesian prior $\D$ has constant-size support. Finally, we show that even in communication-constrained signaling with known valuations, approximating welfare or revenue to a factor better than $1-1/e$ is $NP$-hard.

In a \emph{bipartite signaling problem}, the set of valid signaling maps is represented explicitly as a bipartite graph with $\Omega$ on the left hand side, $\S$ on the right hand side, and a set of edges $E \sse \Omega \cross \S$, as well as an integer $k(=2^b)$. A signaling map $f \in \S^\Omega$ is valid if $(j, f(j)) \in E$ for each $j \in \Omega$, and moreover $|f(\Omega)| \leq k=2^b$ (i.e., the signals can be communicated with at most $b$ bits). We refer to the edge set $E$ as the \emph{bipartite graph constraint}, which limits the compatible signals with each item, and the integer $k$ as the communication constraint, which limits the total number of signals used. To ensure the existence of at least one feasible signaling map and simplify our results, we assume that in instances with $k < |\S|$, there exists a ``no information'' signal $s_0 \in \S$ such that $(j,s_0) \in E$ for all items $j \in \Omega$.

\subsection{Welfare Maximization with Known Valuations}
\label{sec:bipartite_welfare_known}
In the known valuations setting, we show the existence of a polynomial-time, $e/(e-1)$-approximation algorithm for welfare maximization in bipartite signaling. Formally, we prove the following result.

\begin{theorem}\label{thm:bipartite_welfare}
  For bipartite signaling with known valuations, there is a randomized, polynomial-time, $e/(e-1)$-approximation algorithm for  computing the welfare-maximizing signaling scheme.
\end{theorem}

As a warmup, in Section \ref{sec:warmup}, we show that the special case of bipartite signaling without the communication constraint reduces almost directly to combinatorial auctions with XOS valuations. We then extend the result to the general case in Sections \ref{sec:reducemap} and \ref{sec:optimizemap} through a non-trivial reduction to submodular function maximization subject to a matroid constraint.


\subsubsection{\bf A Special Case: Without the Communication Constraint}
\label{sec:warmup}
 We show that  welfare maximization in bipartite signaling with known valuations and without a communication constraint --- i.e. with $k=|\S|$ --- reduces to welfare-maximization in combinatorial auctions\footnote{In \emph{combinatorial auctions}, there is a set $M$ of items and a set $N$ of players. Each player is equipped with a \emph{valuation} mapping subsets of M to the real numbers. Welfare maximization in combinatorial auctions is the problem of assigning the items to the players in order to maximize the sum of the players' values for their assigned bundle of items.} with explicitly-represented XOS valuations,\footnote{A set function $f: 2^X \to \RR$ is called XOS if it can be written as the maximum of additive set functions. Specifically, if $f(A) = \max_{i=1}^k \sum_{j \in A} w_{ij}$ for some integer $k$, and weights $w_{ij} \in \RR$ for $i \in [k]$ and $j \in X$. We say an XOS function is represented explicitly if the weights $w_{ij}$ are given as input.} which can be approximated to within a factor of $\frac{e}{e-1}$ as shown by Dobzinski and Schapira \cite{DS06}.

We are given a set of items $\Omega$, a distribution $p \in \Delta_\Omega$ on items, a set of signals $\S$, a bipartite graph $E \sse \Omega \cross \S$, and known valuations $v_1,\ldots,v_n \in \RR^\Omega$. As in Section \ref{sec:prelim}, we use $\hat{v}_i(j) = p_j v_i(j)$ as shorthand. A valid deterministic signaling scheme $f: \Omega \to \S$ partitions the items $\Omega$ among signals $\S$ --- we let $B_s=f^{-1}(s)$ denote the ``bundle'' of items mapped to signal $s \in \S$. Appealing to Equation \eqref{eq:welfare}, the welfare of $f$ can then be written as:
\[\sum_{s \in \S}  \max_{i=1}^n  \sum_{j \in B_s} \hat{v}_i(j). \]
This can be extended to arbitrary partitions of the items --- i.e. partitions not necessarily respecting the bipartite graph --- by letting $w_{sij}= \hat{v}_i(j)$ for $(j,s) \in E$, and $w_{sij}=0$ for $(j,s) \not\in E$, and defining:
\[ welfare(B)= \sum_{s \in \S}  \max_{i=1}^n  \sum_{j \in B_s} w_{sij}. \]
Since the weights $w_{sij}$ do not reward assignments of items to signals not respecting the bipartite graph, finding the welfare-maximizing deterministic signaling scheme reduces to finding a partition $B$ of items among signals maximizing $welfare(B)$.

We observe that this is an instance of combinatorial auctions with XOS valuations. Namely, if we interpret the signals $\S$ as the ``bidders'' in combinatorial auctions, and $B_s$ as the ``bundle'' of items assigned to $s$, the valuation function of $s$ is simply the XOS function $f(A) = \max_{i=1}^n \sum_{j \in A} w_{sij}$. Welfare maximization in combinatorial auctions, when players have XOS valuations written explicitly, admits an $e/(e-1)$ approximation algorithm that runs in polynomial time, as shown in \cite{DS06}. When combined with Lemma \ref{lem:det_welfare}, this proves Theorem \ref{thm:bipartite_welfare} in the absence of a communication-bounded constraint.

\subsubsection{\bf The General Case: Reduction to Optimization over Mappings}
\label{sec:reducemap}
We now consider the bipartite signaling problem with a cardinality constraint $k$ on the number of signals used. As should be clear from Section \ref{sec:warmup}, the general case of bipartite signaling reduces to a generalization of combinatorial auctions with XOS valuations --- namely, with the additional constraint that at most $k$ players win any items in the combinatorial auction. We are not aware of previous work on this problem, and therefore design an $e/(e-1)$ algorithm for bipartite signaling directly. We break our proof in two parts: first, in this section we show that computing an approximately welfare-maximizing scheme reduces to ``guessing'' the mapping from signals to winning players, then in Section \ref{sec:optimizemap} we show how to express optimization over these mappings as submodular function maximization subject to a \emph{matroid}\footnote{Recall a \emph{matroid} (e.g.~\cite{oxley})
 is a ground set~$X$
 and a non-empty collection~$\I \sse 2^X$ of \emph{independent sets} such that: (i) whenever $S$ is independent and $T \sse S$, $T$ is also
 independent; (ii) whenever $S,T \in \I$ with $|T| < |S|$, there is
 some $x \in S \sm T$ such that $T \cup\ \{x\} \in \I$.}
 constraint, which admits an $e/(e-1)$ approximation algorithm by the result of \cite{vondrak08}.

   We begin by observing that every signaling scheme induces a many-to-one mapping $w:\S \to [n] \union \set{*}$, which maps each signal to the winning player given that signal, where $*$ denotes an unused signal. For determinstic signalling schemes that are valid for our instance of bipartite signalling,  $w(s)=*$ for all but at most $k$ signals $s$.  We call such mappings $w$ \emph{feasible winner mappings}. We reduce the problem of computing a near optimal signalling scheme to an optimization problem over feasible winner  mappings $w$.

   Given a  ``guess'' for the feasible mapping $w$ associated with the welfare-maximizing signalling scheme, computing a deterministic optimal signalling scheme is trivial: each item $j \in \Omega$ is mapped to a winning player who likes it most subject to respecting the bipartite graph $E$. Formally, for every mapping $w: \S \to [n] \union \set{*}$ where $w(s) \neq *$ for at most $k$ signals, we let  $f_w: \Omega \to \S$ be a deterministic signalling scheme satisfying
 \[f_w(j) \in \argmax_{s: (j,s) \in E, w(s) \neq *}  v_{w(s)}(j)\] for items $j$ where such a signal $s$ exists, and $f_w(j)$ is the ``no information'' signal $s_0 \in \S$ otherwise (recall that $s_0$ is a valid signal for all items, see  Section \ref{sec:prelim}). The social welfare of  signalling scheme $f_w$ is, by appealing to Equation \eqref{eq:welfare}, at least
\[welfare(w)= \sum_{s : w(s) \neq *}  \ \  \sum_{j \in f_w^{-1}(s)} \hat{v}_{w(s)}(j). \] Reversing the order of summation, and observing that $f_w$ maps $j$ to a signal maximizing $v_{w(s)}(j)$ subject to respecting the bipartite graph, gives the following equivalent expression for $welfare(w)$:

\begin{equation}\label{eq:welfaremap}
welfare(w)=\sum_{j \in \Omega} \left( \max_{s: (j,s) \in E, w(s) \neq *} \hat{v}_{w(s)}(j) \right).
\end{equation}
The social welfare of $f_w$ is at least $welfare(w)$. Note, however, that $f_w$ may be infeasible, in that it may use up to $k+1$ signals including $s_0$ --- however, in that case $w(s_0)= *$, and therefore a feasible signalling scheme with welfare at least $welfare(w)$ can be gotten by simply choosing an arbitrary signal $s \neq s_0$ with $w(s) \neq *$,  and reassigning all items $f_w^{-1}(s)$ to the ``no information'' signal $s_0$.  Moreover, when  $w$ is the mapping associated with the welfare-maximizing valid signalling scheme, $f_w$ is a welfare maximizing valid signalling scheme with welfare exactly equal to $welfare(w)$. Therefore, finding an approximately welfare-maximizing signalling scheme reduces to finding a mapping $w$ approximately maximizing $welfare(w)$. We summarize this in  the following Lemma.

\begin{lemma}\label{lem:welfaremap}
  Fix $\alpha > 0$. Computing an $\alpha$-approximately welfare maximizing bipartite signalling scheme reduces, in polynomial time, to computing a feasible winner mapping $w: \S \to [n] \union \set{*}$ which $\alpha$-approximately maximizes $welfare(w)$, as given in Equation \eqref{eq:welfaremap}.
\end{lemma}

\subsubsection{\bf The General Case: Optimization over Feasible Winner Mappings}
\label{sec:optimizemap}
We now show how to find a feasible winner mapping $w: \S \to [n] \union \set{*}$ which approximately maximizes $welfare(w)$, as given in Equation \eqref{eq:welfaremap}. We do so by posing this as a submodular function\footnote{A function~$f:2^U \rightarrow \RR$ defined on all subsets of a finite
non-empty set~$U$ is \emph{submodular} if

\[f(S \cup \{i\}) - f(S) \le f(T \cup \{i\}) - f(T)\]

for every $T \sse S \subset U$ and $i \notin S$.
} maximization problem subject to a \emph{truncated partition matroid}\footnote{A \emph{truncated partition matroid} (see \cite{oxley}) is a matroid $(U,I)$ with the following structure. There exists a partition $U_1,\ldots,U_m$ of ground set $U$, integers $k_1,\ldots,k_m$, and an integer $k$, such that a subset $S$ of $U$ is independent, i.e. is in $I$, if and only if $|S \intersect U_j| \leq k_j$ for all $j\in{1,\ldots,m}$, and moreover $|S| \leq k$.} constraint, and invoking the result of   Vondrak \cite{vondrak08} which shows a polynomial-time $e/(e-1)$-approximation algorithm for submodular function maximization subject to an arbitrary matroid.

 Recall that we defined a winner mapping $w$ as feasible if $w(s)\neq *$ for at most $k$ signals $s \in \S$. In order to pose our optimization problem as a constrained submodular maximization problem, we first relax the set of feasible mappings as follows. We consider many-to-many mappings $W \sse \S \cross [n]$ from signals to players, and define the welfare of such a mapping $W$ as follows:
\begin{equation}\label{eq:welfaremultimap}
welfare(W)=\sum_{j \in \Omega} \left( \max_{i,s: (j,s) \in E, (s,i) \in W} \hat{v}_i(j) \right).
\end{equation}
In other words, each item $j$ may be  assigned to a player $i$ so long as there is a signal $s$ that is valid for $j$, and $i$ as one of the ``winning'' players of $s$ as given by $W$; the welfare-maximizing such assignment is used to calculate $welfare(W)$.

When $W$ is a many-to-one mapping --- i.e. assigns to each signal $s$ at most one winning player $w(s)$ --- it is easy to verify that $welfare(W)$ (Equation \eqref{eq:welfaremultimap}) is equal to $welfare(w)$ (Equation \eqref{eq:welfaremap}). Moreover, many-to-one maps $W$ satisfying $|W| \leq k$ are in one-to-one correspondance with the set of feasible winner mappings $w: \S \to [n] \union \set{*}$, where feasibility is as defined in Section \ref{sec:reducemap}.  It is simple to verify that the family of subsets $W$ of $\S \cross [n]$ satisfying those two constraints is a \emph{truncated partition matroid} on ground set $\S \cross [n]$.

We now show that $welfare(W)$ is submodular. We define for each $j \in \Omega$ and $(s,i) \in [S] \cross [n]$ a weight $w_{sij} = \hat{v}_i(j)$ when $(j,s) \in E$, and $w_{sij}=0$ otherwise. We then rewrite $welfare(W)$ as follows.
\begin{equation}
welfare(W)=\sum_{j \in \Omega} \max_{(s,i) \in W} w_{sij}.
\end{equation}
Observe that $welfare(W)$ is the sum of $|\Omega|$ set functions, namely the functions $welfare_j(W) = \max_{(s,i) \in W} w_{sij}$ for all $j \in \Omega$. It is known, and easy to verify, that any set function $f: 2^X \to \RR$ that simply associates a fixed weight $w_x$ with each element $x \in X$, and evaluates to $f(S)= \max_{x \in S} w_x$, is submodular. Therefore, by extension, $welfare_j(W)$ is submodular for each $j$. It is also well known that submodular functions are closed under summation. Therefore, we conclude that $welfare(W)$ is submodular.

Since $welfare(W)$ is submodular, and moreover the set of $W$ corresponding to feasible winner mappings is a matroid, invoking the the result of Vondrak \cite{vondrak08} yields the followin Lemma.

\begin{lemma}
  There is a polynomial-time $e/(e-1)$ approximation algorithm for computing a feasible winner mapping $w$ maximizing $welfare(w)$, as given in Equation \eqref{eq:welfaremap}.
\end{lemma}
\noindent Combined with Lemma \ref{lem:welfaremap}, this completes the proof of Theorem \ref{thm:bipartite_welfare}.



\subsection{Revenue Maximization with Known Valuations}
\label{sec:revenue}
Next we show that, in the case of known valuations, there is a polynomial-time, constant factor approximation algorithm for revenue maximization in cardinality constrained signaling. Our algorithm (Algorithm \ref{alg:cardinality}) simply chooses the best of two signaling schemes, computed  via procedures \ref{alg:cardinality1} and \ref{alg:cardinality2}. We use OPTR and OPTW to denote the maximum revenue and welfare, respectively, of a communication-constrained signaling scheme,  and $v_i(*) = \ex_{j \sim p} v_{i}(j)$ to denote  player $i$'s value for a random item. Moreover, we let $v^* = max_i v_i(*)$ denote the maximum value of a player for a random item, and let $i^*$ be the player attaining this maximum. We note that $v^*$ upper-bounds the contribution of any individual player to the social welfare of any signaling scheme.

\begin{algorithm}
\caption{Algorithm for Communication-constrained Signaling}
\label{alg:cardinality}
\begin{algorithmic}[1]
\INPUT Instance of cardinality constrained signaling, given by $\Omega$, $k$, valuations $v_1,\ldots,v_n$
\OUTPUT Signaling scheme $x^*$
\STATE Run Procedure \ref{alg:cardinality1} to compute deterministic signaling scheme $g$.
\STATE Run Procedure \ref{alg:cardinality2}  to compute signaling scheme $x$.
\STATE Let $x^*$ be the revenue-maximizing signaling scheme among $g$ and $x$
\end{algorithmic}
\end{algorithm}

 \floatname{algorithm}{Procedure}

\begin{algorithm}
\caption{First Sub-procedure for Communication-constrained Signaling}
\label{alg:cardinality1}
\begin{algorithmic}[1]
\INPUT Instance of cardinality constrained signaling, given by $\Omega$, $k$, valuations $v_1,\ldots,v_n$
\OUTPUT Deterministic signaling scheme $g: \Omega \to [k]$
\STATE Compute a deterministic signaling scheme which approximately maximizes welfare (Section \ref{sec:bipartite_welfare_known}). \STATE Repeatedly merge pairs of signals with the same winner, until each signal  has a unique winner. Call the resulting signaling scheme $f$.
\STATE  Sort the signals in decreasing order of their contribution to the social welfare of $f$: $s_1,\ldots,s_k$
\STATE Let $g$ be the signaling scheme which merges signals $s_i$ and $s_{i+1}$ in $f$, for all odd $i$.
 \end{algorithmic}
\end{algorithm}

\begin{algorithm}
\caption{Second Sub-procedure for Communication-constrained Signaling}
\label{alg:cardinality2}
\begin{algorithmic}[1]
\INPUT Instance of cardinality constrained signaling, given by $\Omega$, $k$, and prior $\D$ over valuations $v_1,\ldots,v_n$
\OUTPUT Signaling scheme $x$
\STATE Exclude $i^*$, then compute a deterministic signaling scheme $h$ approximately maximizing welfare for other players (Section \ref{sec:bipartite_welfare_known}). Let $welfare(h)$  denote its welfare (excluding $i^*$), and let $v_s = \max_{i \neq i^*} \sum_{j \in h^{-1}(s)}  \hat{v}_i(j)$ denote the contribution of signal $s$ to the welfare.
\STATE Let $\alpha=v^*/welfare(h)$.
\STATE Let $y$ be the signaling scheme which ignores the realization of the item, and outputs signal $s$ with probability $\alpha \frac{v_s}{v^*}$. (It is easy to verify the probabilities sum to $1$).
\STATE Let $x$ be the signaling scheme which with probability $0.5$ runs $h$, and with remaining probability $0.5$ runs $y$.
 \end{algorithmic}
\end{algorithm}

 \floatname{algorithm}{Algorithm}

We now provide some intuition for our algorithm. The signaling scheme computed by Procedure \ref{alg:cardinality1} guarantees near-optimal revenue  when no individual player's contribution to the optimal social welfare is too large. An approximately welfare-maximizing signaling scheme $f$ is computed as described in Section \ref{sec:bipartite_welfare_known}, then signals are sorted in decreasing order of their contribution to the social welfare, and then pairs of signals are ``merged'' in that order. Formally, merging two signals $s$ and $t$ in a deterministic signaling scheme $f: \Omega \to \S$ gives a new signaling scheme $g$ with $g(j)=\set{s,t}$ whenever $f(j) =s$ or $f(j)=t$, and $g(j)=f(j)$ otherwise.  Merging two signals $s$ and $t$ is tantamount to forcing the two winners of the auction in each of the signals to compete, extracting the value of at least one of them as revenue. Therefore, merging pairs of signals in order of their contribution to the social welfare extracts half the welfare of all but the most valuable signal in $f$. When no individual player contributes much to the social welfare, this is a constant-factor of the welfare of $f$.

The signaling scheme computed by Procedure \ref{alg:cardinality2}, on the other hand, guarantees near-optimal revenue when a single player accounts for a large fraction of the optimal social welfare. In this case, $i^*$'s value for a random item, namely $v^*$, is on the order of the optimal social welfare. Procedure \ref{alg:cardinality2} first computes a deterministic signaling scheme $h$ which $e/(e-1)$-approximately maximizes welfare for players other than $i^*$. Then $h$ is ``mixed'' with a signaling scheme $y$ which releases no information --- i.e. outputs a random signal independent of the realization of the item. Here, mixing two signaling scheme is defined as running each with equal probability. When the probability of each signal $s$ in $y$ is proportional to the contribution of $s$ to the welfare of $h$, player $i^*$'s value conditioned on $s$ is on the order of the value of the welfare of $h$ conditioned on $s$.  Therefore, $i^*$ serves as a price-setting player for all signals, extracting a constant fraction of the social welfare of $h$ in the process. Since, by Lemma \ref{lem:rev_ub}, the welfare of $h$ is a $(1-1/e)$-approximation to the maximum revenue of any signaling scheme, this yields a constant approximation to the optimal revenue.

We now present a formal proof of the approximation ratio of our algorithm. The result follows from two Lemmas.
\begin{lemma}\label{lem:cardinality1}
 Fix $\beta >0$. If $v^* \leq \beta OPTW$, then signaling scheme $g$ --- as computed by procedure \ref{alg:cardinality1} --- has revenue at least $\frac{ (1-1/e - \beta)}{2} OPTW$.
\end{lemma}
\begin{proof}
  Note that $welfare(f) \geq (1-1/e) OPTW$. It is an easy observation that Step 2, which merges signals with the same winning player, preserves the welfare of $f$.

Our assumption that any player's value for a random item is at most $\beta OPTW$ implies that signal $s_1$ accounts for no more than $\beta OPTW$ of the welfare of $f$ --- formally $\max_i \sum_{j \in f^{-1}(s)}  \hat{v}_i(j) \leq v^* \leq \beta OPTW$. Therefore, signals other than $s_1$ account for at least $(1-1/e - \beta) OPTW$ welfare in $f$.

Now, recall that $g$ simply merges pairs of signals in $f$. Given two signals $s$ and $t$ in $f$, merged into a new signal $\set{s,t}$ in $g$, and an arbitrary player $i$, we observe that \[ \sum_{j \in g^{-1}(\set{s,t})}  \hat{v}_i(j) =   \sum_{j \in f^{-1}(s)} \hat{v}_i(j) + \sum_{j \in f^{-1}(t)}  \hat{v}_i(j). \] Intuitively, the conditional value of player $i$ for the new merged signal, weighted by the probability of the signal, is equal to the sum of $i$'s  weighted conditional values for both the component signals.
This implies that, given a merged signal $\set{s,t}$ in $g$, each of the (distinct) winners of $s$ and $t$, which we denote by $i_s$ and $i_t$, have at least as much (weighted) value for $\set{s,t}$ as they did for $s$ and $t$ individually. Since the player with the greatest weighted conditional value for a signal wins, and pays the second-greatest weighted value in expectation, we conclude that the contribution of $\set{s,t}$ to the revenue in $g$ is at least the minimum of the contributions of $s$ and $t$ to the welfare of $f$.

Recalling that we sorted the signals in $f$ in decreasing order of their contribution to the social welfare, and then merged them pairwise in that order, we conclude that the revenue of $g$ is at least half the welfare of $f$ after discarding $s_1$. Our bound on the contribution of $s_1$ to the welfare of $f$ then completes the proof.
\end{proof}

\begin{lemma}\label{lem:cardinality2}
 If $v^* > \beta OPTW$, then $x$ --- as computed by procedure \ref{alg:cardinality2} --- has revenue at least $\frac{\beta}{2} (1-1/e) OPTR$.
\end{lemma}

\begin{proof}[of Lemma \ref{lem:cardinality2}]
  First, observe that $\alpha$, as stated in procedure \ref{alg:cardinality2}, is at least $\beta$. Moreover, Lemma \ref{lem:rev_ub}, combined with the result of Theorem \ref{thm:bipartite_welfare}, shows that $welfare(h) \geq (1-1/e) OPTR$. It remains to show that $x$ extracts at least a $\alpha/2$ fraction of the welfare of $h$ in revenue.

Let $i_s$ denote the winning player of signal $s$ in $h$ (after excluding $i^*$). Recall that $v_s$, as defined in procedure \ref{alg:cardinality2}, denotes the value of player $i_s$ conditioned on signal $s$, weighted by the probability of the signal. Now consider signaling scheme $x$, which mixes --- in equal measure -- signaling scheme $h$ with the scheme which outputs $s$ with probability $\alpha \frac{v_s}{v^*}$. The (weighted) value of $i_s$ for signal $s$ in $x$ is at least half what it was in $h$ --- namely $v_s / 2$. Moreover, the (weighted) value of $i^*$ for signal $s$ is at least  his value for a random item times the probability that signal $s$ is output without regard to the realized item --- doing the calculation, this is $v^* \cdot \frac{1}{2} \cdot \alpha \frac{v_s}{v^*} = \alpha v_s /2$ .

Now, in $x$ both $i^*$ and $i_s$ have (weighted) value for signal $s$ equal to at least $\alpha v_s / 2$. Therefore, the contribution of $s$ to the revenue is at least this amount.  Summing over all signals, the total revenue of $x$ is at least $\frac{\alpha}{2} welfare(h)$. This completes the proof.
\end{proof}

\noindent Setting $\beta=\frac{e-1}{2e-1}$ in Lemmas \ref{lem:cardinality1} and \ref{lem:cardinality2} proves the following theorem.
\begin{theorem}
  Algorithm \ref{alg:cardinality} computes a $\frac{2 e (2e - 1)}{(e-1)^2} \approx 8.17$ approximation to the optimal revenue in the communication-constrained signaling problem.
\end{theorem}

\begin{note}
  We note that optimizing Procedure \ref{alg:cardinality2} to mix $h$ and $y$ with probabilities $\gamma=\alpha/(1+\alpha)$ and $(1-\gamma)$, instead of 0.5 each, improves the approximation ratio in our theorem to $6.47$. We omit the details.
\end{note}

\subsection{Unknown Valuations}
We largely leave open the polynomial-time approximability of welfare and revenue maximization in the unknown valuations setting. Even in the case of welfare maximization in communication-constrained signaling, it appears challenging to design approximation algorithms with run-time scaling sub-exponentially in the support size of $\D$, the distribution over player valuation profiles. Nevertheless, we present here an extension of the result of Section \ref{sec:warmup}  to distributions $\D$ with constant size support. The runtime of the algorithm scales exponentially in the size $t$ of the support. Moreover, we establish a connection of this problem to a generalization of the combinatorial auctions problem used in Section \ref{sec:warmup}.

We assume the distribution $\D$ over valuation profiles is listed explicitly as a set of matrices $\set{v^1, \ldots, v^t}$ where $v^\ell \in \RR^{[n] \cross \Omega}$, and associated probabilities $q_1,\ldots,q_t$. Using a similar derivation to that in Section \ref{sec:bipartite_welfare_known}, the welfare of a deterministic signaling scheme $f$ can be written as
\[ \sum_{s \in \S}  \sum_{\ell=1}^t q_\ell \max_{i=1}^n  \sum_{j \in f^{-1}(s)} \hat{v}^\ell_i(j). \]
Recall that, as described in Section \ref{sec:warmup}, the variant with $t=1$ can be interpreted as an instance of combinatorial auctions with players corresponding to $\S$, who are equipped with XOS valuations. Similarly, the variant with general $t$ can be interpreted as an instance of combinatorial auctions where each player's valuation is a sum of XOS functions. We are not aware of any non-trivial polynomial-time algorithms for this problem that run in time polynomial in $t$ in the computational complexity model. However, since XOS valuations, and therefore their sums, are \emph{subadditive}, the algorithm of Feige \cite{Feige06} guarantees a $2$-approximation if players can answer \emph{demand queries}. Unfortunately we can show that, unlike for XOS functions, answering a demand query for a sum of XOS functions, even approximately to within a constant factor, is NP-hard --- the proof is deferred to the full version of this paper. However, this does not rule out other approaches to our problem not involving demand queries.

We now show that a $1-1/e$ approximation is possible in time exponential in $t$, and polynomial in all other parameters of an instance. We rewrite the welfare of a signaling scheme $f$ as follows.
\[ \sum_{s \in \S} \max_{i_1,i_2,\ldots,i_t}    \sum_{j \in f^{-1}(s)} \sum_{\ell=1}^t q_\ell \hat{v}^\ell_i(j). \]
This, again, is an instance of combinatorial auctions with XOS valuations, though each XOS function is the maximum of $n^t$ additive functions. The $e/(e-1)$-approximation algorithm of \cite{DS06} then runs in time polynomial in $n^t$.

\subsection{Hardness of Approximation}
We now show that our approximation ratio for welfare-maximization in bipartite signaling is tight, even for the special case of communication-constrained signaling and known valuations. As a corollary, the same hardness of approximation result holds for revenue maximization. We use an approximation-preserving reduction from the APX-hard problem max-cover. An instance of max-cover is given by a ground set $[m]$,  a family $A_1,\ldots,A_n$ of subsets of $[m]$, and an integer $k$. The goal is to find $k$ sets from $A_1,\ldots,A_n$ whose union is as large as possible.

Given an instance of max cover, we construct an instance of communication-constrained signaling with known valuations as follows. We let the set $\Omega$ of items be $[m]$, and associate with each set $A_i$ a player $i$ with valuation $v_i:[m] \to \set{0,1}$ such that $v_i(j)=1$ if and only if $j \in A_i$. Moreover, we let $k$ be the constraint on the number of allowable signals, and let $p$ be the uniform distribution over $\Omega$.

Given a solution $A_{i_1},\ldots,A_{i_k}$ for max-cover covering $m'$ items, we show a signaling scheme with welfare at least $m'/m$. We assign items in $A_{i1}$ to signal 1, then assign items in $A_{i_2} \sm A_{i_1}$ to signal 2, and so on until signal $k$. We observe that, for each signal $\ell \in [k]$, there is at least one player -- in particular player $i_\ell$ --- who values at least $|A_{i\ell} \sm \union_{\ell' < \ell} A_{i\ell'} |$ of the items assigned to signal $\ell$. This implies that the welfare of the signaling scheme is at least $\frac{1}{m}|\union_{\ell=1}^k A_{i_{\ell}}| = m'/m$.

Conversely, given a signaling scheme with welfare $m'/m$, we show a solution to max-cover with coverage of at least $m'$. Each signaling scheme partitions the items into $k$ subsets $B_1,\ldots,B_k$, each of which --- say $B_\ell$ --- is associated with a winning player $i_\ell$. The welfare of the signaling scheme is the sum, over all winning players $i_\ell$, of the number of items in $B_\ell$ valued by $i_\ell$, scaled by the per-item probability of $\frac{1}{m}$. Formally, the welfare is $\frac{1}{m} \sum_{\ell=1}^k |A_{i_\ell} \intersect B_\ell|$. Letting $m'= \sum_{\ell=1}^k |A_{i_\ell} \intersect B_\ell|$, it is clear from the fact that $B_1,\ldots,B_k$ is a partition of the items that $|\union_{\ell=1}^k A_{i_\ell}| \geq m'$, as needed to show a solution to max-cover with coverage at least $m'$.

The above reduction, coupled with the hardness of approximation result for max-cover given in \cite{maxcoverhard}, yields the following thereom.
\begin{theorem}\label{thm:welfare-hard}
  There is no polynomial-time $c$-approximation algorithm for welfare-maximization in communication-constrained signaling with known valuations, for any constant $c < \frac{e}{e-1}$, unless $P=NP$.
\end{theorem}

The hardness result of Theorem \ref{thm:welfare-hard} also holds for revenue maximization. This follows from a simple reduction  from welfare maximization to revenue maximization: given an instance of cardinality constrained signaling, produce a new instance whose welfare and revenue are equal by simply duplicating each player.

\begin{corollary}\label {cor:revenue-hard}
  There is no polynomial-time $c$-approximation algorithm for revenue-maximization in communication-constrained signaling with known valuations, for any constant $c < \frac{e}{e-1}$, unless $P=NP$.
\end{corollary}


 \section{Conclusions and Future Work}

Our results initiate the study of signaling in constrained settings. Whereas we obtain preliminary positive and negative results for some natural problems in this setting, we leave open a rich selection of algorithmic problems. We conclude with a statement of several of these open questions.

We leave several open questions in the bipartite signaling setting. Is there a constant-factor approximation algorithm for revenue maximization in this setting with known valuations? What about welfare or revenue in the \emph{unknown} valuation setting, for which we obtain no nontrivial guarantees? 

More generally, we leave open structural questions in constrained signalling problems more generally. We showed that there always exists a deterministic welfare-maximizing signaling scheme. What about the relative power of deterministic and randomized signaling schemes for \emph{revenue maximization}? It was shown in \cite{EmekFGLT12} that in unconstrained settings, there always exists a revenue-maximizing scheme with at least half the optimal welfare. Does such a tradeoff hold in constrained signaling problems?

Finally, there are many other natural signaling problems one may consider. For example, what if products are given as points in a high dimensional hypercube or a high dimensional euclidean space, and an auctioneer must signal a subset of the coordinates? Problems of this form appear related to deep questions in learning theory, such as learning $k$-juntas and others.

\bibliography{signaling}
\bibliographystyle{alpha}

 \appendix

\section{Proofs and Discussion from Section \ref{sec:prelim}}
\label{app:structural}
\subsection{Convexity and Determinism}
Given a constrained signaling problem, we observe that welfare is a convex function of the marginal probabilities $x(j,s)$, and conclude that the welfare-maximizing  constrained signaling scheme is deterministic. Formally, for every valid signaling map $f \in \F \sse \S^\Omega$, we associate a vector $x^f \in \set{0,1}^{\Omega \cross \S}$ where $x^f(j,s)=1$ if and only if $f(j)=s$. A signaling scheme is then associated with a vector $x$ in the convex hull of $\set{x^f : f \in \F}$, and has welfare as given in Equation \ref{eq:welfare}. We observe $welfare(x)$ is a convex function of $x$, and  invoke the following fact:
\begin{fact}
 Let $\P$ be a polytope in Euclidean space, and let $g$ be a convex function. The maximum of $g$ over $\P$ is attained at a vertex of $\P$.
\end{fact}
\noindent We conclude that there is a welfare-maximizing deterministic signaling scheme.
\begin{lemma}\label{app:lem:det_welfare}
  For any constrained signaling problem with unknown valuations, there is valid deterministic signaling scheme which maximizes expected welfare.
\end{lemma}

\subsection{Communication Constraints and Welfare}
When the number of items is finite, we observe a simple bound the number of signals needed to recover the maximum possible social welfare in the known valuations model.
\begin{fact}\label{app:fact:boundknown}
  Consider an $n$ player and $m$ item signaling problem with known valuations. There is a signaling scheme with at most $\min(m,n)$ signals, and welfare equal to that of the optimal unconstrained scheme.
\end{fact}
This follows from the fact that that the scheme that announces the identity of the item, and the scheme that announces the identity of the player who values the item most, are both optimal.

Next, we observe that imposing a communication constraint of $k$ signals reduces the expected welfare of the optimal scheme by a factor of $k/\ell$, when $\ell$ is the number of signals used in the optimal scheme. Invoking Fact \ref{app:fact:boundknown}, this implies that  the best $k$-signal scheme recovers at least a $k/\min(n,m)$ fraction of the welfare of the best unconstrained scheme in the known valuations model.

\begin{lemma}\label{app:lem:scalingsignals}
  Consider an $n$-player and $m$-item signaling problem of unknown valuations. For every $\ell$-signal scheme $x$ and $k \leq \ell$, there is a $k$-signal scheme $y$ such that $welfare(y) \geq \frac{k}{\ell} welfare(x)$.
\end{lemma}
Lemma \ref{lem:scalingsignals} follows easily from Equation \eqref{eq:welfare}, which expresses the welfare of a signaling scheme as the sum of the contributions of all its signals. To see this, let $\S'$ be the set of $k$ signals with the greatest contribution to $welfare(x)$, and let $y$ be any scheme with signal set $\S'$, satisfying $y(j) = x(j)$ whenever $x(j) \in \S'$.

\subsection{Relating Revenue and Welfare}
We now mention a useful, though elementary, upper bound on the optimal revenue achievable via a signaling scheme. This bound follows immediately from the fact that the revenue of a second price auction is at most the welfare of the same auction after excluding an arbitrary player, though we present a proof here for completeness.
\begin{lemma}
  \label{app:lem:rev_ub}
Fix an arbitrary constrained signaling problem with unknown valuations. Let $i'$ be an arbitrary player. The revenue of the revenue-optimal signaling scheme is at most the welfare of the welfare-optimal signaling scheme for all players other than $i'$.
\end{lemma}
\begin{proof}
First, we observe that the revenue can be expressed as
\begin{align}
rev(x,v) &:= \sum_{s \in \S} x(s) \maxtwo_{i=1}^n v_i | s,x \notag \\ &= \sum_{s \in \S}  \maxtwo_{i=1}^n  \sum_{j \in \Omega} x(j,s) \hat{v}_i(j), \label{app:eq:revenue}
\end{align}
where $\maxtwo$ denotes selecting the second largest value. The expected  revenue of the auction over draws of the players' valuations, which we denote by  $rev(x)$, is the expectation of $rev(x,v)$  over $v \sim \D$.

It suffices to show that $rev(x,v)$ (equation \eqref{app:eq:revenue}) is at most $welfare(x,v_{-i'})$ (equation \eqref{eq:welfare}).
  \begin{align*}
    rev(x,v) &= \sum_{s \in \S}  \maxtwo_{i=1}^n  \sum_{j \in \Omega} x(j,s) \hat{v}_i(j) \\
    &\leq \sum_{s \in \S}  \max_{i\neq i'}  \sum_{j \in \Omega} x(j,s) \hat{v}_i(j) \\
&= welfare(x,v_{-i'})
  \end{align*}
\end{proof}



\end{document}